\documentclass[12pt]{amsart}

\usepackage{mathpple}       
\usepackage{amssymb,amsmath,amsthm,amstext,graphics,amsfonts,hyperref, color}
\usepackage{cite}

\newcommand{\FF}{{\mathbb{F}}}
\newcommand{\seq}{\subseteq}
\newcommand{\ZZ}{\mathbb{Z}}

\DeclareMathOperator{\wtg}{wt}
\DeclareMathOperator{\wtgr}{Wt}
\DeclareMathOperator{\cwe}{cwe}
\DeclareMathOperator{\swe}{swe}
\DeclareMathOperator{\swc}{swc}
\DeclareMathOperator{\crt}{CRT}
\DeclareMathOperator{\rk}{rank}
\DeclareMathOperator{\frk}{f-rank}
\DeclareMathOperator{\su}{supp}
\DeclareMathOperator{\ho}{Hom}
\DeclareMathOperator{\inc}{inc}
\DeclareMathOperator{\res}{res}

\newtheorem{theorem}{Theorem}[]
\newtheorem{lem}[theorem]{Lemma}
\newtheorem{cor}[theorem]{Corollary}
\newtheorem{prop}[theorem]{Proposition}
\theoremstyle{definition}
\newtheorem{ex}[theorem]{Example}
\newtheorem{definition}{Definition}[]
\newtheorem{remark}{Remark}

\title{Structure  of linear codes over the ring $B_k$}
\author{Irwansyah}
\address{Mathematics Department\\
Universitas Mataram\\
Jl. Majapahit 62, Mataram\\
INDONESIA}

\author{Djoko Suprijanto}
\address{Combinatorial Mathematics Research Group\\
Faculty of Mathematics and Natural Sciences\\
Institut Teknologi Bandung\\
Jl. Ganesha 10, Bandung, 40132\\
INDONESIA}
\date{}

\begin{document}

\maketitle

\begin{abstract}
We study the structure of linear codes over the ring $B_k$ which is defined by
$\FF_{p^r}[v_1,v_2,\ldots,v_k]/\langle v_i^2=v_i,~v_iv_j=v_jv_i \rangle_{i,j=1}^k.$
In order to study the codes, we begin with studying the structure of the ring $B_k$
via a Gray map which also induces a relation between codes over $B_k$ and codes
over $\FF_{p^r}.$
We consider Euclidean and Hermitian self-dual codes, MacWilliams relations, as well as Singleton-type bounds
for these codes.
Further, we characterize cyclic and quasi-cyclic codes using their images under the Gray map,
and give the generators for these type of codes.

\vspace{0.25cm}

\\
\end{abstract}

\section{Introduction}

"Coding theory arose in the twentieth century as a problem in engineering concerning
the efficient transmission of information. $\cdots$ Specifically, the theory was developed so
that electronic information could be transmitted and stored without error. Electronic
information can generally be thought of as a series of ones and zeros. Therefore,
coding theory, from this perspective, was largely done using the binary field as the
alphabet. However, the alphabet was quickly generalized to finite fields, at least for
mathematicians, since many of the proofs and techniques were identical to the binary
case viewed as the field with two elements. This type of coding theory remains a
vital part of electrical engineering in terms of ensuring effective communication in
telephones, computers, television, and the internet."\footnote{\cite{doug2017}, page 1.}
 \cite{doug2017}

As mentioned above, in the beginning, algebraic coding theory considers finite fields as alphabet codes.
Codes over finite rings were introduced later in early 1970s by Blake \cite{Blake1, Blake2}.
He \cite{Blake1} showed how to construct codes over $\ZZ_m$ from cyclic codes
over $\FF_p,$ where $p$ is a prime factor of $m.$  He \cite{Blake2} then further observed
the structure of codes over $\ZZ_{p^r}.$  Spiegel \cite{Spiegel1, Spiegel2} generalized Blake's results
to codes over $\ZZ_m,$ where $m$ is an arbitrary positive integer.
Study of codes over finite rings attracted great interest in algebraic coding theory
through the work of Hammons, Kumar, Calderbank, Sloane and Sol\'{e} \cite{hammons},
where they show how several well-known families of nonlinear binary codes were
intimately related to linear codes over $\ZZ_4.$  
Since Hammons et.al. \cite{hammons}
many people have been considering codes over various finite rings.  Among the recent works are
\cite{Chatouh17, li17} where they considered codes over the ring
$\FF_2[u_1,u_2,\ldots,u_q] / \langle u_i^2=0,u_iu_j=u_ju_i \rangle_{i,j=1}^q$ for $q \geq 1$ and
$\ZZ_4+u\ZZ_4+v\ZZ_4+uv\ZZ_4,$ respectively.


Recently, codes over infinitely family of rings
which are an extension of the binary field,
$A_k:=\mathbb{F}_2[v_1,v_2,\ldots,v_k]/ \langle v_i^2=v_i,~v_iv_j=v_jv_i\rangle_{i,j=1}^k,$
are considered by Cengellenmis, Dertli, and Dougherty \cite{dougherty-ceng}.
The reason why they \cite{dougherty-ceng} considered linear codes over these rings is,
among other thing, because they have two
Gray maps which relate codes over such rings and binary codes.
These rings have also non-trivial automorphisms which nicely be used to define skew-cyclic codes over the rings
\cite{irw16b} (see also \cite{abualrub}  for a special case) and
produced certain optimal self-dual cyclic codes. Very recently, investigations into
skew-cyclic codes over the ring $A_k$ have been generalized by the authors and others (\cite{irw17}) to the
ring $B_k:=\mathbb{F}_{p^r}[v_1,v_2,\ldots,v_k]/\langle v_i^2=v_i,~v_iv_j=v_jv_i\rangle_{i,j=1}^k.$
However, the investigation into the structure of
linear codes over the ring $B_k$ is yet to be done.

The purpose of this paper is to study some structural aspects of linear codes over the ring $B_k.$
After preliminaries to the ring $B_k$ in Section 2, Section 3 considers linear codes over $B_k.$
MacWilliams relations are also provided here.
Self-duality of the codes over $B_k$ is presented in Section 4.  Singleton-type bounds for the codes
are mentioned in section 5.  The paper is ended
by characterizing cyclic and quasi-cyclic codes using their images under the Gray map.
Throughout this paper, we follow standard definitions of undefined terms as used in many coding theory books
(e.g. \cite{huffman}).


\section{Preliminaries}

In this section we study the ring $B_k,$ its units, ideal structure and its properties.  Some
of the results here have appeared in \cite{irw17b}, but we include here for the reader's convenience.
We begin by defining an infinite family of the ring $B_k$ as a generalization of the ring $A_k.$

Let $v_i$, for $1 \leq i \leq k,$ be an indeterminate and $\mathbb{F}_{q}$ be
a finite field of order $q.$ The ring $B_k$ is the ring of the form
\[
B_k=\mathbb{F}_{p^r}[v_1,v_2,\ldots,v_k]/\langle v_iv_j=v_jv_i,~v_i^2=v_i\rangle_{i,j=1}^k,
\]
for some prime $p$ and non-negative integer $r.$ It is a finite non-chain ring as there exist
more than one maximal ideals.  For example, if $k=1,$ then
$B_1=\FF_{p^r}+v\FF_{p^r},$ where $v^2=v.$  We also define $B_0:=\FF_{p^r}.$
The ring $B_k$ forms a commutative algebra over the prime field $\mathbb{F}_{p^r}.$

Let $[1,m]:=\{1,2,\ldots,m\},$ $\Omega=\{1,2,\dots,k\},$ and $2^\Omega$ be the collection of all subsets of $\Omega.$
Also, let $w_i$ be an element in the set $\{v_i,1-v_i\},$ for $1\leq i\leq k.$
It has been proven (see \cite[Lemma 1]{irw17b}) that the ring $B_k$ can be viewed
as a vector space over $\FF_{p^r}$ of dimension $2^k$ whose basis consists of
elements of the form $\prod_{i \in S} w_i,$ where $S \in 2^\Omega.$
Moreover, the following three properties are discussed in \cite{irw17b}.

\begin{prop}
Let $I=\langle \alpha_1,\dots,\alpha_m\rangle$ be an ideal in $B_k,$ for some $\alpha_1,\dots,\alpha_m\in B_k.$ Then
\[
I=\left \langle\sum_{\substack{A\subseteq [1,m],\\ A\not=\emptyset}}
(-1)^{|A|+1} \left( \prod_{j\in A}\alpha_j \right)^{p^r-1} \right \rangle.
\]
\label{genideal}
\end{prop}

\begin{prop}\label{formmaxideal}
An ideal $I$ in $B_k$ is maximal if and only if $I=\langle w_1,w_2,\dots,w_k\rangle,$
where $w_i \in \{v_i, 1-v_i\}$ for $1 \leq i \leq k.$
Moreover, direct sum of any two of these ideals is equal to $B_k.$
\end{prop}

\begin{lem}
The ring $B_k$ has characteristic $p$ and cardinality $(p^r)^{2^k}.$
\end{lem}

We also note that every element $a$ in $B_k$ can be written as
\[
a=\sum_{S\in 2^\Omega}\alpha_Sv_S
\]
for some $\alpha_S\in\mathbb{F}_{p^r},$ where $v_S=\prod_{i\in S}v_i,$ and $v_{\emptyset}:=1.$

Following \cite{irw17b}, define a \emph{Gray map} $\varphi$ as follows
\[
\begin{array}{llcll}
   \varphi & : & B_k & \longrightarrow & \mathbb{F}_{p^r}^{2^k}, \\
           &   & a=\sum_{i=1}^{2^k}\alpha_{S_i} v_{S_i} & \longmapsto &
           (\sum_{S\subseteq S_1}\alpha_S,\sum_{S\subseteq S_2}\alpha_S,\ldots,
		    \sum_{S\subseteq S_{2^k}}\alpha_{S}).
\end{array}
\]
The map $\varphi$ is bijective. Furthermore, this map can be extended into $n$ tuples of $B_k$
naturally as follows
\[
\begin{array}{llcll}
   \overline{\varphi} & : & B_k^n & \longrightarrow & \mathbb{F}_{p^r}^{n2^k}, \\
           &   & (a_1,a_2,\ldots,a_n) & \longmapsto & (\varphi(a_1),\varphi(a_2),\ldots,\varphi(a_n)).
\end{array}
\]
Since $\varphi$ is a bijection map, so is $\overline{\varphi}.$

Using the characterization of maximal ideals as given in Proposition \ref{formmaxideal},
we obtain the following theorem.

\begin{theorem}
The ring $B_k$ is isomorphic via Chinese Remainder Theorem to $\mathbb{F}_{p^r}^{2^k}.$
\label{crt}
\end{theorem}

\begin{proof}
Similar to the proof of \cite[Theorem 2.5]{dougherty-ceng} using Proposition~\ref{formmaxideal}.
\end{proof}

We can also define another Gray map as discussed in \cite{irw17b}.
Every element $\gamma$ in $B_j,$ where $j\geq 1,$ can be written as $\gamma=\alpha + \beta v_j,$
for some $\alpha, \beta \in B_{j-1}$ (considering $B_0:=\FF_{p^r}$). Then, for $j \geq 2,$
let $\phi_j:B_j \rightarrow B_{j-1}^2$ be a map defined by
\[
\phi_j(\alpha + \beta v_j) = (\alpha, \alpha + \beta).
\]

Now, we define another \emph{Gray map} $\Phi_j : B_j \rightarrow \mathbb{F}_{p^r}^{2^j},$
by $\Phi_1(\gamma) = \phi_1(\gamma),$ $\Phi_2(\gamma) = \phi_1 \circ \phi_2(\gamma)$
and
\[
\Phi_j (\gamma) = \phi_1 \circ \phi_2 \circ  \dots \circ \phi_{j-2} \circ \phi_{j-1} \circ \phi_j (\gamma),
\text{ for } 3 \leq j \leq k.
\]
The map $\Phi_k$ is a bijective map as shown in \cite{irw17b}.
Note that, this map is also a generalization of the Gray map in \cite{dougherty-ceng}.
Again, this map $\Phi_j$ can be extended to $n$ tuples of $B_k$
naturally becomes $\overline{\Phi}_j$ as before.

\section{Linear Codes over $B_k$}

A code $C \seq B_k^n$ is called \emph{linear} if $C$ is a submodule over $B_k.$
In next section
we consider the important class of linear codes, namely  self-dual codes over $B_k.$
For that purpose we define the inner product in $B_k^n.$

\subsection{Inner Product}

We consider two inner products in $B_k^n.$ First, we have a \emph{Euclidean inner-product} defined as
\[
[\mathbf{u},\mathbf{u}']=\sum_{i=1}^{n}u_iu_i',
\]
for any $\mathbf{u}=(u_1,\dots,u_{n})$ and $\mathbf{u}'=(u_1',\dots,u_{n}')$ in $B_k^n.$
The second inner-product, called \emph{Hermitian inner-product,} is defined as
\[
[\mathbf{u},\mathbf{u}']_H=\sum_{i=1}^nu_i\overline{u_i'},
\]
where $\overline{u_i'}=1-u_i'.$

Now, define $C^\bot=\{\mathbf{u} \in B_k^n:~[\mathbf{u},\mathbf{u}']=0,\forall\mathbf{u}'\in C\}$
and $C^H=\{\mathbf{u} \in B_k^n:~[\mathbf{u},\mathbf{u}']_H=0,\forall\mathbf{u}'\in C\}$
to be \emph{Euclidean dual} and \emph{Hermitian dual} of $C,$ respectively.
We note that, since $B_k$ is a principal ideal ring, then it is a Frobenius ring.
Therefore, we have $|C||C^\bot|=|C||C^H|=|B_k|^n=p^{r2^kn}$ (see \cite[Corollary 3.2]{doug2017}).

Define the function  $\gamma: B_k \longrightarrow \{-1,0\}$ as follows
\begin{equation}
\gamma(\alpha_Sv_S)=
\begin{cases}
-1, & \text{if}\;\alpha_S\not=0,\\
0, & \text{otherwise}.
\end{cases}
\label{charcoef}
\end{equation}

Then, we have the following properties.

\begin{prop}
If $I=\left\langle \sum_{S\in 2^\Omega}\alpha_Sv_S\right\rangle$ is an ideal in $B_k,$ then
\[
\begin{aligned}
I^\bot &=\left\langle 1+\left(\gamma(\alpha_{S_1}v_{S_1})v_{S_1}
+\gamma(\alpha_{S_2}v_{S_2})v_{S_2}+\prod_{i=1}^2\gamma(\alpha_{S_i}v_{S_i})v_{S_1\cup S_2}
\right. \right. \\
& \quad +\left. \left. \cdots+\prod_{i=1}^{2^k}\gamma(\alpha_{S_i}v_{S_i})v_{S_1\cup\cdots\cup S_{2^k}}\right)\right\rangle
\end{aligned}
\]
and
\[
\begin{aligned}
I^H &=\left\langle \overline{1+\left(\gamma(\alpha_{S_1}v_{S_1})v_{S_1}
+\gamma(\alpha_{S_2}v_{S_2})v_{S_2}+\prod_{i=1}^2\gamma(\alpha_{S_i}v_{S_i})v_{S_1\cup S_2}
\right.} \right. \\
&  \quad \left. \overline{\left.
+\dots+\prod_{i=1}^{2^k}\gamma(\alpha_{S_i}v_{S_i})v_{S_1\cup\cdots\cup S_{2^k}}\right)}\right\rangle.
\end{aligned}
\]
\label{idealdual}
\end{prop}

\begin{proof}
We will show, by mathematical induction, that
\[
\left(\sum_{S\in 2^\Omega}\alpha_Sv_S\right)\left(1-v_{S_1}-v_{S_2}+v_{S_1\cup S_2}+\dots+v_{\Omega}\right)=0.
\]
First, we note that
\[
\begin{aligned}
\left(\alpha_{S_1}v_{S_1}+\alpha_{S_2}v_{S_2}\right)\left(1-v_{S_1}-v_{S_2}+v_{S_1\cup S_2}\right)
& =  \alpha_{S_1}v_{S_1}+\alpha_{S_2}v_{S_2} \\
 &  -\alpha_{S_1}v_{S_1}-\alpha_{S_1}v_{S_1\cup S_2}+\alpha_{S_1}v_{S_1\cup S_2}\\
 &  -\alpha_{S_2}v_{S_1\cup S_2}-\alpha_{S_2}v_{S_2}+\alpha_{S_2}v_{S_1\cup S_2}\\
 & =  0.
\end{aligned}
\]
Now, assume that,
\[
\begin{array}{l}
\left(\alpha_{S_1}v_{S_1}+\alpha_{S_2}v_{S_2}+\dots+\alpha_{S_{n-1}}v_{S_{n-1}}\right)\times  \\
\left(1-v_{S_1}-v_{S_2}-\dots-v_{S_{n-1}}+v_{S_1\cup S_2}+\dots+v_{S_{n-1}\cup S_{n-2}}
+\dots+(-1)^{n-1}v_{S_1\cup\cdots\cup S_{n-1}}\right)=0.
\end{array}
\]
Then, we will show
\[
\begin{array}{l}
\left(\alpha_{S_1}v_{S_1}+\alpha_{S_2}v_{S_2}+\dots+\alpha_{S_{n}}v_{S_{n}}\right)\times  \\
\left(1-v_{S_1}-v_{S_2}-\dots-v_{S_{n-1}}+v_{S_1\cup S_2}+\dots+v_{S_{n-1}\cup S_{n-2}}+\dots \right.\\
\left. +(-1)^{n-1}v_{S_1\cup\cdots\cup S_{n-1}}\right)(1-v_{S_n})=0.
\end{array}
\]
Let
\[
\alpha = \alpha_{S_1}v_{S_1}+\alpha_{S_2}v_{S_2}+\dots+\alpha_{S_{n-1}}v_{S_{n-1}},
\]
and
\[
\beta = 1-v_{S_1}-v_{S_2}-\dots-v_{S_{n-1}}+v_{S_1\cup S_2}+\dots+v_{S_{n-1}\cup S_{n-2}}
+\cdots+(-1)^{n-1}v_{S_1\cup\cdots\cup S_{n-1}}.
\]
The previous multiplication can be simplified as
\[
\begin{aligned}
(\alpha+\alpha_{S_n}v_{S_n})\beta (1-v_{S_n})& =  \beta \alpha_{S_n}v_{S_n}(1-v_{S_n})\\
 & =  0.\\
\end{aligned}
\]
Note that, the second equality comes from the assumption. Moreover,
\[
\beta (1-v\alpha_{S_n})=1-v_{S_1}-v_{S_2}-\dots-v_{S_{n}}+v_{S_1\cup S_2}
+\cdots+v_{S_{n-1}\cup S_{n}}+\dots+(-1)^{n}v_{S_1\cup\cdots\cup S_{n}}.
\]
Therefore, by the above arguments, we have
\[
\left(\sum_{S\in 2^\Omega}\alpha_Sv_S\right)\left(1-v_{S_1}-v_{S_2}+v_{S_1\cup S_2}+\dots+v_{\Omega}\right)=0.
\]
The rest of the proof is similar to the proof of \cite[Theorem 4.4]{dougherty-ceng}.
\end{proof}

%

\subsection{Minimal Generating Sets}

Based on the form of maximal ideals in the ring $B_k$ as described in Proposition~\ref{formmaxideal},
it is easy to see that, there are $2^k$ maximal ideals in $B_k.$ Let $\mathcal{I}_i$
be the maximal ideal as in Proposition~\ref{formmaxideal}, where $1\leq i\leq 2^k.$
Also, the direct sum of any two distinct $\mathcal{I}_i$ and $\mathcal{I}_{j}$ will produce $B_k,$
since if for some $l,$ $v_l\in \mathcal{I}_i,$ then $1-v_l$ must be in $\mathcal{I}_{j},$
for $j \neq i.$

Now, consider the map
\[
\Theta_i : B_k \longrightarrow B_k/\mathcal{I}_i,~\beta \longmapsto \beta \pmod{\mathcal{I}_i}, ~\text{ for }1 \leq i \leq 2^k.
\]
It is well-known that since $\mathcal{I}_i$ is a maximal ideal, then $B_k/\mathcal{I}_i$ is a field.

We define the map
\[
\Theta : B_k \longrightarrow B_k/\mathcal{I}_1\times B_k/\mathcal{I}_2\times \cdots \times B_k/\mathcal{I}_{2^k}~
\]
by $\Theta(\beta)=\left(\Theta_1(\beta),\dots,\Theta_{2^k}(\beta)\right).$
Then, the map $\Theta^{-1}$ is an isomorphism by Chinese Remainder Theorem.

Following \cite{dougherty-ceng}, we define the following notions.

\begin{definition}
Let $\mathbf{u}_1,\mathbf{u}_2,\dots,\mathbf{u}_k$ be vectors in $B_k^n.$
Then $\mathbf{u}_1,\mathbf{u}_2,\dots,\mathbf{u}_k$ are \emph{independent}
if $\sum_{i=1}^k\alpha_i\mathbf{u}_i=\mathbf{0}$ implies that $\alpha_i\mathbf{u}_i=\mathbf{0}$ for all $i \in[1,k].$
\label{defind}
\end{definition}

\begin{definition}
The vectors $\mathbf{u}_1,\mathbf{u}_2,\dots,\mathbf{u}_k$ in $B_k^n$ are \emph{modular independent} if vectors \[
\Theta_i(\mathbf{u}_1),\Theta_i(\mathbf{u}_2),\dots,\Theta_i(\mathbf{u}_k)
\]
are linearly independent for some $i.$
\label{defmodind}
\end{definition}

Following the notion in \cite{doug-liu}, the generating set that is both independent
and modular independent is called {\it minimal generating set}.

It is not always possible to put the minimal generating set of a code into a matrix in standard form. See the following example.
\begin{ex}
Take $k=1, p=2,$ and $r=2.$ Let $\alpha$ be the root of the polynomial
$f(x)=x^2+x+1\in \mathbb{F}_2[x],$ and $\mathbb{F}_{2^2}=\mathbb{F}_2[\alpha].$
Now, let $C$ be a code generated by the vector $(1+v,v,v)$ over $B_1=\mathbb{F}_4+v\mathbb{F}_4,$
where $v^2=v.$ As we can see, $C=\{(0,0),(1+v,v,v),\alpha(1+v,v,v),(1+\alpha)(1+v,v,v)\}.$
The vector $(1+v,v,v)$ is both independent and modular independent, since the condition
$\beta(1+v,v,v)=(0,0,0)$ if and only if $\beta(1+v,v,v)=(0,0,0),$
which satisfies the one in Definition~\ref{defind},
and $\Theta_1(1+v,v,v)=(1,0,0)$ is linearly independent over $\mathbb{F}_4.$

Recall that the generator matrix in standard form of a linear code over $B_1$ is
\[
\begin{pmatrix}
I_{k_1} & A_1 & A_2 & A_3 \\
0 & (1+v)I_{k_2} & A_4 & A_5 \\
0 & 0 & vI_{k_3} & A_6 \\
\end{pmatrix}.
\]

To make a vector $(1+v,v,v)$ fit to the above form, we need two other vectors,
{i.e.} $(1+v)(1,0,0)$ and $v(0,1,1).$ The previous vectors are not modular independent,
since their image under $\Theta_1$ are $(1,0,0),~(0,0,0),$ and their image under $\Theta_2$ are $(0,0,0),~(0,1,1).$
\end{ex}

Let $\mathbf{u}=(\beta_1,\beta_2,\dots,\beta_n)$ be a nonzero vector in $B_k^n,$
and $\langle \beta_1,\dots,\beta_n\rangle$ be an ideal generated by $\beta_1,\dots,\beta_n.$
Also, let $I(\mathbf{u})=|\langle \beta_1,\dots,\beta_n\rangle|.$
Then we have the following result as a generalization of \cite[Theorem 4.3]{dougherty-ceng}.

\begin{prop}
If $C$ is a code with minimal generating set $\mathbf{u}_1,\dots,\mathbf{u}_s,$ then $|C|=\prod_{i=1}^sI(\mathbf{u}_i).$
\label{codeid}
\end{prop}

\begin{proof}
Similar to the proof of \cite[Theorem 4.3]{dougherty-ceng}.
\end{proof}

\subsection{MacWilliams Relations} \label{Macwill}

MacWilliams relation provides a connection between  weight distribution of a linear code and its dual.
In this subsection, we study some classes of MacWilliams relation of linear codes over $B_k.$

Recall that, the field $\mathbb{F}_{p^r}$ can be viewed as an $r$-dimensional vector space over $\mathbb{F}_p.$
Let $\mathcal{B}=\{b_0,b_1,\dots,b_{r-1}\}$ be a basis of $\mathbb{F}_{p^r}$ over $\mathbb{F}_{p}.$
For any element $x=\alpha_0b_0+\dots+\alpha_{r-1}b_{r-1}$ define
\[
\wtg(x)=\sum_{i=0}^{r-1}\alpha_i.
\]
Also, for any $\mathbf{x}=(x_1,\dots,x_n)\in\mathbb{F}_{p^r}^n,$ we define the weight of $\mathbf{x}$ as follows
\[
\wtgr(\mathbf{x})=\sum_{i=1}^n\wtg(x_i).
\]
Note that, if $p=2$ and $\mathcal{B}$ is a trace-orthogonal basis, the above weight is the Lee weight in \cite{betsumiya}.
Now, using the Gray map $\varphi,$ we define the weight of any $\alpha\in B_k$ as
\[
\wtgr(\alpha)=\wtgr(\varphi(\alpha)),
\]
which will correspond to Lee weight in \cite{dougherty-ceng} when $p=2$ and $r=1.$

As we have already seen, the ring $B_k$ is isomorphic to $\mathbb{F}_{p^r}^{2^k}.$
Therefore, the generating character for $\widehat{B_k}$ is the product of generating character
for the field $\mathbb{F}_{p^r}.$ Now, we define the character $\chi$ for $\mathbb{F}_{p^r}$ such that
\[
\chi(x)=\xi^{\wtg(x)},
\]
for any $x\in \mathbb{F}_{p^r},$ where $\xi=\exp(2\pi i/p).$ We can see that, $\chi$ is a generating character.
Therefore, the generating character $\chi$ for $B_k$ is
\[
\chi_1(\beta)=\xi^{\wtgr(\varphi(\beta))},
\]
for any $\beta\in B_k,$ by Chinese Remainder Theorem.

Now, define the matrix $T$ indexed by $\alpha,\beta\in B_k,$ as follows
\[
T_{\alpha,\beta}=\chi_\alpha(\beta)=\chi(\alpha\beta),
\]
and the matrix $T_H$ as follows
\[
\left(T_H\right)_{\alpha,\beta}=\chi_\alpha(\overline{\beta})=\chi(\alpha\overline{\beta}).
\]

Define the complete weight enumerator for a code $C$ as
\[
\cwe_C(\mathbf{X})=\sum_{\mathbf{c}\in C}\prod_{b\in B_k}X_b^{n_b(\mathbf{c})},
\]
where $n_b(\mathbf{c})$ is the number of occurrences of the element $b$ in $\mathbf{c}.$
Then, we have the following MacWilliams relations for complete weight enumerator.

\begin{theorem}{(MacWilliams Relations for CWE)}
Let $C$ be a linear code over $B_k,$ then
\begin{equation}
\cwe_{C^\bot}(\mathbf{X})=\frac{1}{|C|}\cwe_C(T\cdot \mathbf{X})
\end{equation}
and
\begin{equation}
\cwe_{C^H}(\mathbf{X})=\frac{1}{|C|}\cwe_C(T_H\cdot \mathbf{X})
\end{equation}
\label{macrel}
\end{theorem}

\begin{proof}
This theorem is a consequence of \cite[Corollary 8.2]{wood}.
\end{proof}

As we can see, $T$ is a $p^{r2^k}$ by $p^{r2^k}$ matrix indexed by the elements of $B_k.$
Denote by $B_k^\times$ the group of units in $B_k.$ Define the relation $\sim$ as follows:
~$\alpha\sim\alpha'$ if $\alpha'=u\alpha,$
for some $u\in B_k^\times.$
It can be seen that the relation $\sim$ is an equivalence relation,
so we define $\mathcal{A}=\{\alpha_1,\dots,\alpha_t\}$ be the set of representatives.
Let $S$ be the $t$ by $t$ matrix indexed by the elements in $\mathcal{A}.$
Also, define $S_{\alpha,\beta}=\sum_{\gamma\sim\beta}T_{\alpha,\gamma}.$

Now, if $\alpha\sim\alpha'$ then for any column $\beta$ we have
\[
S_{\alpha',\beta} = \sum_{\gamma\sim\beta}T_{\alpha',\gamma} = \sum_{\gamma\sim\beta}\xi^{\wtgr(\varphi(\alpha'\gamma))}.
\]
Since $\varphi(\alpha\gamma)=\varphi(\alpha)\varphi(\gamma),$
where the multiplication in the righthand side of equality carried out coordinate-wise, we have
\[
\begin{aligned}
\sum_{\gamma\sim\beta}T_{\alpha',\gamma} & =  \sum_{\gamma\sim\beta}\xi^{\wtgr(\varphi(\alpha) \varphi(u)\varphi(\gamma))}\\
 & = \sum_{\gamma'\sim\beta}\xi^{\wtgr(\varphi(\alpha)\varphi(\gamma'))}\\
 & = \sum_{\gamma'\sim\beta}T_{\alpha,\gamma'}\\
 & = S_{\alpha,\beta}.
\end{aligned}
\]
Therefore, $S_\alpha=S_{\alpha'}$ when $\alpha\sim\alpha'.$

Define the symmetrized weight enumerator for a code $C$ as
\[
\swe_C(\mathbf{Y_{\mathcal{A}}})=\sum_{\mathbf{c}\in C}\prod_{\alpha\in\mathcal{A}}Y_\alpha^{\swc_\alpha(\mathbf{c})},
\]
where $\swc_\alpha(\mathbf{c})=\sum_{\alpha'\sim\alpha}n_{\alpha'}(\mathbf{c}).$ Then, we have the following theorem.

\begin{theorem}{(MacWilliams Relation for SWE)}
Let $C$ be a linear code over $B_k,$ then
\[
\swe_{C^\bot}=\frac{1}{|C|}\swe_C(S\cdot \mathbf{Y}_{\mathcal{A}}).
\]
\label{macrel2}
\end{theorem}

\begin{proof}
Apply \cite[Theorem 8.4]{wood}.
\end{proof}

\begin{remark}
We note that when $p=2$ and $r=1,$ the above symmetrized weight enumerator
is equal to complete weight enumerator, since the only unit in $A_k$ is 1.
\end{remark}

\section{Self-Dual Codes}

We define two types of self-duality for codes as follows.

\begin{definition}
A code $C$ is called \emph{Euclidean self-dual} if $C=C^\bot.$ It is called \emph{Hermitian self-dual} if $C=C^H.$
It is  called \emph{Euclidean [Hermitian] self-orthogonal} code if $C\subseteq C^\bot$ [$C\subseteq C^H$].
\end{definition}


The proposition below shows the non-existence of Euclidean self-dual codes of length $1$ over $B_k.$
Basically, this proposition
generalizes \cite[Theorem 4.2]{dougherty-ceng}.

\begin{prop}
If $I$ is an ideal of $B_k$ then $I^\bot\not= I.$
\label{nosde}
\end{prop}

\begin{proof}
Similar to the proof of Theorem 4.2 in \cite{dougherty-ceng}.
\end{proof}

\begin{remark}
Proposition \ref{nosde} does not hold for the Hermitian inner product.
For example, take $k=3,$ for any $p$ and $r,$ the Hermitian dual of the ideal
$I=\langle v_2\rangle$ is $I^H=\langle \overline{1-v_2}\rangle=\langle v_2\rangle.$
\end{remark}

Following \cite{dougherty-ceng}, we define a map $\prod_{j,k}: B_j\rightarrow B_k^{2^{k-j}}$ for $j>k$ such that
\[
\prod_{j,k}=\phi_{k+1}\circ\phi_{k+2}\circ\cdots\circ\phi_{j}.
\]

Then, we have the following results.

\begin{prop}
If $C$ is an Euclidean [Hermitian] self-dual code over $B_j,$ then $\prod_{j,k}(C)$
is an Euclidean [Hermitian] self-orthogonal code over $B_k$ for $j>k.$
\end{prop}

\begin{proof}
Similar to the proof of \cite[Theorem 5.1]{dougherty-ceng}.
\end{proof}

\begin{prop}
If $\mathbf{c}_1,\mathbf{c}_2,\dots,\mathbf{c}_s$ generate a self-dual code over $B_k,$
then $\mathbf{c}_1,\mathbf{c}_2,\dots,\mathbf{c}_s$ generate a Hermitian [Euclidean] self-dual code over $B_j,$ for $j>k.$
\end{prop}

\begin{proof}
Since $B_k\subseteq B_j,$ for $j>k,$ if $\mathbf{c}_1,\mathbf{c}_2,\dots,\mathbf{c}_s$
generate a self-dual code, say $C_k,$ over $B_k,$ then $\mathbf{c}_1,\mathbf{c}_2,\dots,\mathbf{c}_s$
also generate a self-orthogonal code, say $C_j,$ over $B_j,$ for $j>k.$ Therefore, we have $C_j\subseteq C_j^H.$
The rest of the proof is similar to the proof \cite[Theorem 5.2]{dougherty-ceng} by using Proposition~\ref{codeid}.
\end{proof}

As a direct consequence, we have the following corollary.

\begin{cor}
If $C$ is a Hermitian [Euclidean] self-dual code over $B_k,$
then there exists a self-dual code $C'$ over $B_j,$ for all $j>k,$ with $\prod_{j,k}(C')=C.$
\end{cor}

\subsection{Euclidean Self-Dual Codes}

The following theorem gives a characterization for Euclidean self-dual codes over $B_k.$  Note
that here $\crt(C_1,\ldots, C_{2^k})$ is defined as a pre-image of $(C_1,\ldots, C_{2^k})$ under the map $\Theta,$ namely
\[
\crt(C_1,\ldots, C_{2^k}):=\{\Theta^{-1}(x_1,x_2,\ldots,x_{2^k}):~x_i \in C_i,~ 1 \leq i \leq 2^k\}.
\]

\begin{theorem}
A code $C$ is an Euclidean self-dual code if and only if $C=\crt(C_1,\ldots, C_{2^k})$
and each $C_i$ is an Euclidean self-dual code over $\mathbb{F}_{p^r}.$
\label{chareuc}
\end{theorem}

\begin{proof}
Apply Theorem~\ref{crt} and \cite[Theorem 2.3]{dougherty}.
\end{proof}

Then, we have the following consequence.

\begin{cor}
Euclidean self-dual codes of length $n$ over $B_k$ exist if and only if $n$ is even.
\end{cor}

\begin{proof}
Apply Theorem~\ref{chareuc} and the well-known fact that the Euclidean self-dual codes of length $n$
over finite field exist if and only if $n$ is even.
\end{proof}

The theorem below describes the relation between Euclidean self-duality of the code and
Euclidean self-duality of its image under the Gray maps $\overline{\varphi}$ and $\Phi_k.$

\begin{theorem}
The image under the maps $\overline{\varphi}$ and $\Phi_k$ of an Euclidean self-dual code is an Euclidean
self-dual code over finite field $\mathbb{F}_{p^r}.$
\end{theorem}

\begin{proof}
Follows from the fact that $\overline{\varphi}$ and $\Phi_k$ are linear maps.
\end{proof}

\begin{definition}
An Euclidean self-dual code is said to be \emph{Type II} if and only if the weights of every codewords,
i.e. the element in the code, are divisible by $4.$
\end{definition}

Regarding the Type II codes, we have the following theorem.

\begin{theorem}
If $C$ is a Type II Euclidean self-dual code then $C=\crt(C_1,\ldots,C_{2^k})$ and
each $C_i$ is a Type II code over $\mathbb{F}_{p^r}.$
\end{theorem}

\begin{proof}
This follows from the definition of weight for codewords over $B_k.$
\end{proof}

\subsection{Hermitian Self-Dual Codes}

The following theorem gives two Hermitian self-dual codes of length 1 over $B_k.$

\begin{theorem}
The code $I=\langle v_i\rangle$ and $J=\langle 1-v_i\rangle$ are Hermitian self-dual codes of length 1.
\label{herm}
\end{theorem}

\begin{proof}
Apply Proposition~\ref{idealdual} and the fact that $\overline{v_i}=1-v_i$ and $\overline{1-v_i}=v_i.$
\end{proof}

As a direct consequence, we know the existence of Hermitian self-dual codes for
all lengths.

\begin{cor}
Hermitian self-dual codes over $B_k$ exist for all lengths.
\end{cor}

\begin{proof}
Take the direct products of codes in Theorem~\ref{herm}.
\end{proof}

The image of a Hermitian self-dual code need not to be self-dual. Consider the following example.

\begin{ex}
Let $p=2, r=2,$ and $k=1.$ Take $I=\langle 1-v \rangle.$ As we can see, $\varphi(1-v)=(1\;0),$
which is neither Euclidean [Hermitian] self-orthogonal nor Euclidean [Hermitian] self-dual.
\end{ex}

By similar point of view as in \cite{dougherty-ceng}, we have that $B_j$ is isomorphic to $B_{j-1}^2$
(as a convention, $B_0=\mathbb{F}_{p^r}$) via the Chinese Remainder Theorem, for any $j\geq 1.$
As a consequence, if $C$ is a Hermitian self-dual code over $B_j,$ then $C$ is isomorphic to $D\times D^\bot,$
where $D$ is any code over $B_{k-1}.$ Then, we have the following theorem.

\begin{theorem}
If $C$ is a Hermitian self-dual code over $B_k,$ then, with the proper arrangement of indices, $C$ is isomorphic to
\[
C_1\times C_1^\bot\times C_2\times C_2^\bot\times\cdots\times C_{2^{k-1}}\times C_{2^{k-1}}^\bot,
\]
where $C_i$ is any linear code over $\mathbb{F}_{p^r}.$
\label{hermisom}
\end{theorem}

\begin{proof}
Use the above fact inductively and rearrange the images.
\end{proof}

\begin{theorem}
If $C$ is a Hermitian self-dual code of length $n$ over $B_k,$ then $\Phi_k(C)$ is
a formally self-dual code of length $2^kn$ over $\mathbb{F}_{p^r}$ with respect to the Hamming weight.
\label{fsd}
\end{theorem}

\begin{proof}
From the facts that $\Phi_k$ is a distance preserving map and the Hamming weight enumerator for
codes over $B_k$ satisfies the MacWilliams relation as in \cite{irw16} (see Section \ref{Macwill}),
we have the Hamming weight enumerator for $\Phi_k(C)$ also satisfies the MacWilliams relation.
\end{proof}

Then, we have the following construction for formally self-dual codes over $\mathbb{F}_{p^r}$ with respect to the Hamming weight.

\begin{cor}
If $C_1,C_2,\dots,C_{2^{k-1}}$ are arbitrary codes over $\mathbb{F}_{p^r},$ then
\[
\Phi_k(\crt(C_1,C_1^\bot,\ldots,C_{2^{k-1}},C_{2^{k-1}}^\bot))
\]
(given the right ordering of indices) is a formally self-dual code of length $n$ over $\mathbb{F}_{p^r}$
with respect to the Hamming weight.
\end{cor}

\begin{proof}
Apply Theorem~\ref{hermisom} and Theorem~\ref{fsd}.
\end{proof}

\section{Singleton-type bounds}

The {\it rank} of a code $C$ is defined as the minimum number of generators of $C,$ and
the {\it free rank} of $C$ is defined as the maximum of the ranks of the free $B_k$-submodule of $C.$

The Singleton bound states that a code $C$ of length $n$ over an
alphabet $A$ satisfies $d_H(C)\leq n-\log_{|A|}(|C|)+1,$ where $d_H(C)$ denotes the Hamming distance of a code $C.$
A code attaining this bound is called \emph{MDS code.}
Meanwhile, it is show in \cite{shiro} that a code $C$ of length $n$ over
a principal ideal ring satisfies $d_H(C)\leq n-r+1,$
where $r$ is the rank of $C.$ A code attaining this bound is called \emph{MDR code.}
Notice that, if $C$ is an MDR and free code, then $C$ is an MDS code,
since the rank and the free rank of $C$ coincide.

Let $C$ be a code over $B_k$ with $C=\crt(C_1,\ldots,C_{2^k}),$ where $C_i$ is a code over $\mathbb{F}_{p^r}.$
As proved in \cite{dougherty2}, we have that

\begin{equation}
|C|=\prod_{i=1}^{2^k}|C_i|,
\end{equation}

\begin{equation}
\rk(C)=\max\{\rk(C_i):~1\leq i\leq 2^k\},
\label{rank}
\end{equation}

\begin{equation}
d_H(\crt(C_1,\ldots,C_{2^k}))=\min\{d_H(C_i):~1\leq i\leq 2^k\},
\label{mindis}
\end{equation}
and $C$ is a free code if and only if each $C_i$ is a free code of the same rank.
Moreover, using the above facts, as stated in \cite[Theorem 6.3]{dougherty2}
we have that if $C_i$ is an MDR code for each $i,$ then $C=\crt(C_1,\ldots,C_{2^k})$ is also an MDR code,
and if $C_i$ is an MDS code of the same rank for each $i,$ then $C=\crt(C_1,\ldots,C_{2^k})$ is also an MDS code.

The following theorem gives a construction for MDS Euclidean self-dual codes over $B_k.$

\begin{theorem}
If $C_1,\dots,C_{2^k}$ are MDS Euclidean self-dual codes over $\mathbb{F}_{p^r},$
with the same rank, then $C=\crt(C_1,\ldots,C_{2^k})$ is an MDS Euclidean self-dual code over $B_k.$
\end{theorem}

\begin{proof}
Apply Theorem~\ref{chareuc} and \cite[Theorem 6.3]{dougherty2}.
\end{proof}

Meanwhile, the following theorem gives us a way to construct MDR codes over $B_k.$

\begin{theorem}
Let $C=\crt(C_1,\ldots,C_{2^k})$ with $C_j$ is an MDS code over $\FF_{p^r},$ for some $j.$ If $\rk(C_i)\leq \rk(C_j)$
for all $i$ and $d_H(C_i)\geq d_H(C_j)$ for all $i,$ then $C$ is an MDR code over $B_k.$
\end{theorem}

\begin{proof}
Simply use the Equations~(\ref{mindis}) and (\ref{rank}).
\end{proof}

If we consider the Lee weight with respect to a basis of $\mathbb{F}_{p^r}$ in the previous section,
then we have the following Singleton-type bound for codes over $\mathbb{F}_{p^r}.$

\begin{lem}
If $C$ is a linear code of length $n$ over $\mathbb{F}_{p^r}$ and the minimum Lee weight of $C$ is $d_L(C),$ then
\[
\displaystyle{\left\lfloor\frac{d_L(C)-1}{r(p-1)}\right\rfloor\leq n-\log_{p^{r}}|C|}.
\]
\label{singtr}
\end{lem}

\begin{proof}
Recall that $|\mathbb{F}_{p^r}|=p^{r}$ and the maximum value of $a_i$ in \cite{shiro} is $r(p-1).$
Then, using \cite[Theorem 1]{shiro} we have the desired inequality.
\end{proof}

Then, by using the above Lemma, we have the following Singleton-type bound for codes over $B_k.$

\begin{theorem}
If $C$ is a linear code of length $n$ over $B_k$ and the minimum weight of $C$ is $d_L(C),$ then
\[
\displaystyle{\left\lfloor\frac{d_L(C)-1}{r(p-1)}\right\rfloor\leq 2^kn-\log_{p^{r}}|C|}.
\]
\label{singtr2}
\end{theorem}

\begin{proof}
This result follows from Lemma~\ref{singtr} and the fact that the Gray image of
a code of length $n$ over $B_k$ is a code of length $2^kn$ over $\mathbb{F}_{p^r.}$
\end{proof}

A code attaining the bound in Theorem~\ref{singtr2} is called \emph{Maximum Lee Distance Separable (MLDS)} code.
For example, let $p=2, r=2,$ and $k=1.$ Take $C=\langle(1,1,\dots,1)\rangle,$ then $|C|=4^{4n}$ and $d_L(C)=n.$
As we can see, $C$ is an MLDS code over $B_1=\mathbb{F}_4+v\mathbb{F}_4,$ where $v^2=v.$

Now, we establish an algebraic version of a Singleton bound.
Let $\rk(C)$ be the rank of $C$ and $\frk(C)$ be the free rank of $C.$ We have the following lemma.

\begin{lem}
If $C$ is a code of length $n$ over $B_k$ then
\[
\rk(C)+\frk(C)=n.
\]
\label{rfrank}
\end{lem}

\begin{proof}
Similar to the proof of \cite[Lemma 7.4]{dougherty-ceng}.
\end{proof}

For any $\mathbf{x}=(x_1,\dots,x_n)\in B_k^n,$ let $\su(\mathbf{x})=\{i:~x_i\not=0\}.$
Also, let $D$ be a $B_k$-submodul of $B_k^n$ and $M\subseteq N:=\{1,2,\dots,n\}.$ Define
\[D(M)=\{\mathbf{x}\in D:~\su(\mathbf{x})\subseteq M\},\]
\[D^*=\ho_{B_k}(D,B_k).\]

We can see that $D(M)=D\cap B_k^n(M)$ is a $B_k$-submodul of $B_k^n$ and $|B_k^n(M)|=(p^r)^{2^k|M|}.$
Also, there exists an isomorphism
\[D^*\cong D.\]

Moreover, there exists a $B_k$-homomorphism as follows,
\[\begin{array}{llll}
g : & B_k^n & \rightarrow & D^*\\
 & y & \mapsto & (\hat{y} : x\mapsto [x,y]).
\end{array}\]
The map $g$ is surjective. Then, we have the following proposition.

\begin{prop}
If $C$ is a code of length $n$ over $B_k$ and $M\subseteq N,$ then there exists an exact sequence of $B_k$-modules :
\[
0\longrightarrow C^\bot(M)\stackrel{\inc}{\longrightarrow}
B_k^n(M)\stackrel{g}{\longrightarrow}C^*\stackrel{\res}{\longrightarrow}C(N-M)^*\longrightarrow 0,
\]
where $\inc$ and $\res$ are the inclusion map and the restriction map, respectively.
\label{exact}
\end{prop}

\begin{proof}
This is a special case of \cite[Lemma 1]{shiro}.
\end{proof}

Since the maximum weight of elements in $\mathbb{F}_{p^r}$ is $r(p-1)$ and every element of $B_k$
is a pre-image of $2^k$ elements of $\mathbb{F}_{p^r},$ we have that for any $\mathbf{x}\in B_k^n,$

\begin{equation}
\wtgr(\mathbf{x})\leq r(p-1)2^k|\su(\mathbf{x})|.
\label{leebound}
\end{equation}

Now, we have the following algebraic version of a Singleton-type bound.

\begin{theorem}
If $C$ is a linear code of length $n$ over $B_k$ with minimum Lee weight $d_L(C),$ then
\begin{equation}
\left\lfloor\frac{d_L(C)-1}{r(p-1)2^k}\right\rfloor\leq n-\rk(C).
\end{equation}
\label{singalg}
\end{theorem}

\begin{proof}
By replacing $C$ in the exact sequence in Lemma~\ref{exact} by $C^\bot,$ we have the following exact sequence,
\begin{equation}
0\longrightarrow C(M)\stackrel{\inc}{\longrightarrow}B_k^n(M)\stackrel{g}
{\longrightarrow}(C^\bot)^*\stackrel{\res}{\longrightarrow}C^\bot(N-M)^*\longrightarrow 0.
\label{exact2}
\end{equation}
Take $M\subseteq N$ such that $|M|=\left\lfloor\frac{d_L(C)-1}{r(p-1)2^k}\right\rfloor.$
Now, by inequality~(\ref{leebound}), for any $\mathbf{x}\in C(M)$ we have that
\[
\wtgr(\mathbf{x})\leq r(p-1)2^k|M|=d_L(C)-1.
\]
The previous equation gives $C(M)^*=0,$ and by isomorphism $C(M)\cong C(M)^*,$ we have $C(M)=\mathbf{0}.$
Now, apply the duality functor $^*=\ho_{B_k}(-,B_k)$ to the exact sequence (\ref{exact2}).
Using the fact that $C(M)^*=0$ and $B_k^n(M)\cong B_k^n(M)^*,$ we have the following short exact sequence,
\[
0\longrightarrow C^\bot(N-M)\longrightarrow C^\bot\longrightarrow B_k^n(M)\longrightarrow 0.
\]
As we know, $B_k^n(M)\cong B_k^{|M|},$ so $B_k^n(M)$ is a projective module.
Therefore, the above short exact sequence is split, which gives
\[
C^\bot\cong C^\bot(N-M)\oplus B_k^n(M).
\]
The previous isomorphism gives
\[
\frk(C^\bot)\geq \frk(B_k^n(M))=|M|=\left\lfloor\frac{d_L(C)-1}{r(p-1)2^k}\right\rfloor.
\]
Then, by Lemma~\ref{rfrank}, we have the desired inequality.
\end{proof}

A code attaining the bound in Theorem~\ref{singalg} is called \emph{Maximum Lee Distance with respect to Rank (MLDR)} code.

\begin{prop}
If $C$ is a free MLDR code with $\left\lfloor\frac{d_L(C)-1}{r(p-1)2^k}\right\rfloor=\frac{d_L(C)-1}{r(p-1)2^k},$
then $C$ is an MLDS code.
\end{prop}

\begin{proof}
By the assumption, we have that
\[d_L(C)=r(p-1)2^kn-r(p-1)2^k\rk(C)+1.\]
If $C$ is a free code, then $|C|=|B_k|^{\rk(C)}=(p^r)^{2^k\rk(C)}.$
Therefore, we have $\log_{p^r}|C|=2^k\rk(C),$ which makes $C$ satisfies the bound in Theorem~\ref{singtr2}.
\end{proof}

%

\section{Cyclic and Quasi-Cyclic Codes over $B_k$}

In this section we will characterize quasi-cyclic and cyclic codes over $B_k$
in terms of their images under the Gray map $\overline{\varphi}.$ We refer to \cite{irw16} to see how the Gray map works.

Let us recall first the definition of cyclic and quasi-cyclic codes.
A linear code $C \seq B_k^n$ is called \emph{cyclic [quasi-cyclic of index $l$]} if it satisfies the following
property (1)[(2)] below:

\begin{itemize}
\item[(1)] $c=(c_0,c_1,\ldots,c_{n-1})\in C$ $\Rightarrow$ $T(c)=(c_{n-1},c_{0},c_{1},\ldots,c_{n-2})\in C.$

\item[(2)] $c=(c_0,c_1,\ldots,c_{n-1})\in C$ $\Rightarrow$
 $T^l(c)=(c_{n-l\pmod n},c_{n-l+1\pmod n},$\\ $c_{n-l+2\pmod n},\ldots,c_{n-1-l\pmod n})\in C.$
\end{itemize}

The following theorem characterizes quasi-cyclic codes over $B_k.$

\begin{theorem}
A code $C$ with length $n$ is quasi-cyclic of index $l$ over $B_k$ if and only if
$C=\overline{\varphi}^{-1}(C_1,C_2,\dots,C_{2^k})$ and each code $C_i$ with length $n$ is quasi-cyclic of length $l$
over $\mathbb{F}_{p^r},$ for $1\leq i\leq 2^k.$
\label{charquasi}
\end{theorem}


\begin{proof}
($\Longrightarrow$)
For any code $C$ of length $n$ over $B_k,$ there exist linear codes $C_1,C_2,\ldots,C_{2^k}$ of length $n$ over $\mathbb{F}_{p^r}$
such that $C=\overline{\varphi}^{-1}(C_1,C_2,\ldots,C_{2^k}).$ Let $\mathbf{a}_i$ be any codeword in $C_i,$
where $1\leq i\leq 2^k.$ Let $\mathbf{c}=(c_0,c_1,\dots,c_{n-1})$ be a codeword in $C$
such that $\mathbf{c}=\overline{\varphi}^{-1}(\mathbf{a}_1,\mathbf{a}_2,\ldots,\mathbf{a}_{2^k}).$  We have that
\[
T^l(\mathbf{c})=\left( c_{n-l\;\text{mod}\;n},c_{n-l+1\;\text{mod}\;n},\dots,c_{n-l-1\;\text{mod}\;n} \right)
\]
is also in $C.$ Let $c_i=\sum_{S\subseteq 2^\Omega}\alpha_S^{(i)}v_S,$ for some $\alpha_S^{(i)}\in \mathbb{F}_{p^r}.$ Consider
\[
\begin{aligned}
\overline{\varphi}(\mathbf{c}) & =
\left(\sum_{S\subseteq S_1}\alpha_S^{(0)}v_S,\ldots,\sum_{S\subseteq S_{2^k}}\alpha_S^{(0)}v_S,
\sum_{S\subseteq S_1}\alpha_S^{(1)}v_S,\ldots,\sum_{S\subseteq S_{2^k}}\alpha_S^{(1)}v_S,\right.\\
& \left.\quad \ldots, \sum_{S\subseteq S_1}\alpha_S^{(n-1)}v_S,\ldots,\sum_{S\subseteq S_{2^k}}\alpha_S^{(n-1)}v_S\right).
\end{aligned}
\]
Notice that,
\[
\mathbf{a}_j=
\left(\sum_{S\subseteq S_j}\alpha_S^{(0)}v_S,\sum_{S\subseteq S_1}\alpha_S^{(1)}v_S,\ldots,
\sum_{S\subseteq S_j}\alpha_S^{(n-1)}v_S\right),
\]
for $1\leq j \leq 2^k.$ Now, consider
\[
\begin{aligned}
\overline{\varphi}(T^l(\mathbf{c})) & =
\left(\sum_{S\subseteq S_1}\alpha_S^{(n-l\;\text{mod}\;n)}v_S,\ldots,\sum_{S\subseteq S_{2^k}}\alpha_S^{(n-l\;\text{mod}\;n)}v_S,
\ldots\right.\\
& \left.\quad \ldots,\sum_{S\subseteq S_1}\alpha_S^{(n-l-1\;\text{mod}\;n)}v_S,\ldots,
\sum_{S\subseteq S_{2^k}}\alpha_S^{(n-l-1\;\text{mod}\;n)}v_S\right),
\end{aligned}
\]
which gives
\[
T^l(\mathbf{a}_j)=\left(\sum_{S\subseteq S_j}\alpha_S^{(n-l\;\text{mod}\;n)}v_S,
\sum_{S\subseteq S_1}\alpha_S^{(n-l+1\;\text{mod}\;n)}v_S,\dots,\sum_{S\subseteq S_j}\alpha_S^{(n-l-1\;\text{mod}\;n)}v_S\right)
\]
is in $C_j$ for all $1\leq j\leq 2^k.$ Therefore, $C_j$ is a quasi-cyclic code of index $l,$
for all $1\leq j\leq 2^k.$\\[0.25cm]($\Longleftarrow$) Simply reverse the previous process.
\end{proof}

Since cyclic codes are just the quasi-cyclic codes of index $l=1,$ we have the following consequence.

\begin{theorem}
A code $C$ is a cyclic code of length $n$ over $B_k$ if and only if $C=\overline{\varphi}^{-1}(C_1,C_2,\dots,C_{2^k})$
and $C_i$ is a cyclic code of length $n$ over $\mathbb{F}_{p^r},$ for all $1\leq i\leq 2^k.$
\label{charcyclic}
\end{theorem}

In terms of polynomial generators, we have the following properties.

\begin{cor}
Let $C=\overline{\varphi}^{-1}(C_1,C_2,\ldots,C_{2^k})$ be a quasi-cyclic code over $B_k,$
where $C_1,C_2,\ldots,C_{2^k}$ are quasi-cyclic codes over $\mathbb{F}_{p^r}.$
If $C_i=\langle g_{1_i}(x),\ldots,g_{m_i}(x)\rangle,$ for all $i=1,\ldots,2^k,$ then
\[
\begin{aligned}
C & =\left \langle v_{S_1}g_{1_1}(x),\ldots,v_{S_{2^k}}g_{1_1}(x),\ldots,v_{S_1}g_{m_1}(x),\ldots,v_{S_{2^k}}g_{m_1}(x),\right.\\
  & \quad \left. \ldots, v_{S_1}g_{m_s}(x),\dots,v_{S_{2^k}}g_{m_s}(x) \right \rangle.
\end{aligned}
\]
\label{gen1}
\end{cor}

\begin{proof}
For any $c(x)\in C,$ there exist $c_i(x)\in C_i,$ where $1\leq i\leq 2^k,$
such that $c(x)=\overline{\varphi}^{-1}(c_1(x),\dots,c_{2^k}(x)).$ Now, let
\[
c_j(x)=\sum_{k=1}^{m_j}\alpha_{kj}(x)g_{jk}
\]
for all $j=1,\dots,2^k,$ then we have
\[
\begin{aligned}
 c(x) & = v_{S_1}\left(\sum_{k=1}^{m_1}\alpha_{k1}(x)g_{1k}\right)\\
      &   \quad +\dots+v_{S_i}\left(\sum_{k=1}^{m_i}\alpha_{ki}(x)g_{ik} -\sum_{S_j\subseteq S_i}
      \left(\sum_{k=1}^{m_j}\alpha_{kj}(x)g_{jk}\right)\right) \\
      &  \quad +\dots+v_{S_{2^k}}\left(\sum_{k=1}^{m_{2^k}}
      \alpha_{k2^k}(x)g_{2^k k}-\sum_{j=1}^s\left(\sum_{k=1}^{m_j}\alpha_{kj}(x)g_{jk}
	  \right)\right)
\end{aligned}
\]
as we hope.
\end{proof}

\begin{cor}
Let $C=\overline{\varphi}^{-1}(C_1,\dots,C_{2^k})$ be a cyclic code over $B_k,$ where $C_1,\dots,C_{2^k}$ are
cyclic codes over $\mathbb{F}_{p^r}.$
If $C_i=\langle g_{i}(x)\rangle,$ for all $i=1,\dots,2^k,$ then
\[
C=\left \langle v_{S_1}g_{1}(x),\dots,v_{S_{2^k}}g_{1}(x),\dots,v_{S_1}g_{2^k}(x),\dots,v_{S_{2^k}}g_{2^k}(x) \right \rangle.
\]
\label{gen2}
\end{cor}

\begin{proof}
This follows from Corollary \ref{gen1}.
\end{proof}

\begin{ex}
Let $\mathbb{F}_4=\mathbb{F}_2[\alpha],$ where $\alpha^2=\alpha+1.$ Also let $B_1=\mathbb{F}_4+v\mathbb{F}_4,$ where $v^2=v,$ and
\[C=\langle (v+\alpha,0,v+\alpha,0)\rangle.\]
We can see that $C$ is a quasi-cyclic code of index $2$ over $B_1.$ Also, we can check that
\[\varphi((v+\alpha,0,v+\alpha,0))=\left(\begin{array}{cccc}
\alpha & 0 & \alpha & 0 \\
\alpha+1 & 0 & \alpha+1 & 0
\end{array}\right).\]
Therefore, if we take $C_1=\langle(\alpha,0,\alpha,0)\rangle$ and $C_2=\langle(\alpha+1,0,\alpha+1,0)\rangle,$
then $C=\varphi^{-1}(C_1,C_2).$ Moreover, $C_1$ and $C_2$ are quasi-cyclic codes of index $2$ over $\mathbb{F}_4.$
\end{ex}

\begin{ex}
Let $\mathbb{F}_4=\mathbb{F}_2[\alpha],$ where $\alpha^2=\alpha+1.$
Also let $B_1=\mathbb{F}_4+v\mathbb{F}_4,$ where $v^2=v,$ and
\[C'=\langle (\alpha v+1,v),(v,\alpha v+1)\rangle.\]
As we can see, $C'$ is a cyclic code over $B_1.$ Also, we have that
\[\varphi((\alpha v+1,v))=\left(\begin{array}{cc}
1 & 0 \\
\alpha+1 & 1
\end{array}\right)\]
and
\[\varphi((v,\alpha v+1))=\left(\begin{array}{cc}
0 & 1 \\
1 & \alpha+1
\end{array}\right).\]

Therefore, if $C_1'=\langle(1,0),(0,1)\rangle$ and
$C_2'=\langle(\alpha+1,1),(1,\alpha+1)\rangle,$ then $C'=\varphi^{-1}(C_1',C_2').$ We can see that,
$C_1'$ and $C_2'$ are cyclic codes over $\mathbb{F}_4.$
\end{ex}

\section{Conclusion}
We have studied linear codes over the ring $B_k$ defined by
$\FF_{p^r}[v_1,v_2,\ldots,v_k]/\langle v_i^2=v_i,~v_iv_j=v_jv_i \rangle_{i,j=1}^k,$
and obtained some results including MacWilliams relations, Singleton-type bounds, and also necessary and sufficient condition
of cyclic and quasi-cyclic codes over the ring.  The questions about concrete examples of extremal codes
(MDR codes, MDS codes, MLDR codes, MLDS codes), defined
to be linear codes attained Singleton-type bounds, are of some important from theoretical as well as practical viewpoints.
Moreover, recently we obtained several results regarding linear codes over the ring $\mathcal{R}_k$ defined
from $B_k$  by changing $\FF_{p^r}$ to arbitrary finite Frobenius ring $R.$
Further detail results which are now in preparation
will be published elsewhere in a separate paper \cite{irw18}.


\section*{Acknowledgement}
This research is supported in part by \emph{Riset ITB 2017.}  A part of this work was done while the second author
visited Research Center for Pure and Applied Mathematics (RCPAM), Graduate School of Information Sciences,
Tohoku University, Japan on July 2017 - August 2017
under the financial support from \emph{Penelitian Unggulan Perguruan Tinggi (PUPT) Kemenristekdikti 2017}.
The second author thanks Prof. Hajime Tanaka for kind hospitality.

\end{document}